\documentclass[journal]{IEEEtran}
\usepackage[T1]{fontenc}
\usepackage[latin9]{inputenc}
\usepackage{verbatim}
\usepackage{float}
\usepackage{amsmath}
\usepackage{amsthm}
\usepackage{amssymb}
\usepackage{esint}
\usepackage[switch,columnwise]{lineno}
\usepackage{relsize}

\usepackage{xcolor}
\usepackage{stackengine}

\newcommand*{\vertbar}{\rule[-1ex]{0.5pt}{3.5ex}}

\newtheorem{lem}{Lemma} 
\newtheorem{thm}{Theorem}

\theoremstyle{definition}
\newtheorem{definition}{Definition}

\theoremstyle{remark}
\newtheorem{rem}{Remark}

\newcommand{\fH}{\mathcal{H}}

\newcommand{\myc}{\gamma}
\newcommand{\myq}{q}

\usepackage{extpfeil}

\newcommand{\bfr}{\mathbf{R}}
\newcommand{\bfrh}{\hat{\mathbf{R}}}

\newcommand{\td}{\tilde{d}}

\newcommand{\bfA}{\mathbf{A}}
\newcommand{\bfy}{\mathbf{y}}

\newcommand{\bX}{\mathbf{X}}
\newcommand{\bx}{\mathbf{x}}
\newcommand{\bZ}{\mathbf{Z}}
\newcommand{\bs}{\mathbf{s}}
\newcommand{\bd}{\mathbf{d}}
\newcommand{\bS}{\mathbf{S}}

\newcommand{\tr}{\mathrm{tr}}

\newcommand{\bD}{\mathbf{D}}
\newcommand{\bU}{\mathbf{U}}

\newcommand{\bu}{\mathbf{u}}
\newcommand{\bv}{\mathbf{v}}

\newcommand{\toas}{\overset{\text{a.s.}}{\to}}

\newcommand{\buni}{\bu_{n,i}}

\newcommand{\bfrs}{\mathbf{R}^*}
\newcommand{\bfrsi}{\mathbf{R}^{*-1}}
\newcommand{\toprob}{ \overset{\mathrm{p}}{\to} }

\newcommand{\by}{\mathbf{y}}
\newcommand{\ba}{\boldsymbol{a}}
\newcommand{\bb}{\boldsymbol{b}}
\DeclareMathOperator*{\argmin}{arg\,min}
\newcommand{\bdeltahat}{\boldsymbol{\hat{\delta}}}
\newcommand{\bdeltacheck}{\boldsymbol{\check{\delta}}}
\newcommand{\bdhat}{\boldsymbol{\hat{d}}}
\newcommand{\bdcheck}{\boldsymbol{\check{d}}}

\ifCLASSINFOpdf
\else
\fi
\begin{document}
%
\title{Space-Time Adaptive Detection at Low Sample Support}  
%
%
%

\author{Benjamin~D.~Robinson,~\IEEEmembership{Member,~IEEE}, 
        Robert~Malinas,~\IEEEmembership{Student Member,~IEEE}, Alfred O. Hero III,~\IEEEmembership{Fellow,~IEEE}
\thanks{This work was generously supported by AFOSR grant 19COR1936 and ARO grant W911NF-15-1-0479}
\thanks{Benjamin Robinson is with Air Force Research Lab.}
\thanks{Robert Malinas is with University of Michigan}
\thanks{Manuscript received April 19, 2005; revised August 26, 2015.}}

\maketitle

\begin{abstract}
%
%
An important problem in space-time adaptive detection is the estimation of the large $p\times p$ interference covariance matrix from training signals.
When the number of training signals $n$ is greater than $2p$, existing estimators are generally considered to be adequate, as demonstrated by fixed-dimensional asymptotics.  But in the low-sample-support regime ($n < 2p$ or even $n < p$) fixed-dimensional asymptotics are no longer applicable.  The remedy undertaken in this paper is to consider the ``large dimensional limit'' in which $n$ and $p$ go to infinity together.
In this asymptotic regime, a new type of estimator is defined (Definition~2), shown to exist (Theorem~1), and shown to be detection-theoretically ideal (Theorem~2). 
Further, asymptotic conditional detection and false-alarm rates of filters formed from this type of estimator are characterized (Theorems~3 and 4) and shown to depend only on data that is given, even for non-Gaussian interference statistics.
The paper concludes with several Monte Carlo simulations that compare the performance of the estimator in Theorem~1 to the predictions of Theorems~2-4, showing in particular higher detection probability than Steiner and Gerlach's Fast Maximum Likelihood estimator.
\end{abstract}

%
\IEEEpeerreviewmaketitle

%
%
%
%

\section{Introduction} \label{sec:intro}
\IEEEPARstart{A}{fundamental} challenge in radar is the multichannel detection of targets embedded in interference consisting of jammers and other non-target scatterers known as clutter.  \emph{Space-time adaptive processing} (STAP) is a technique that amounts to applying an adaptive linear filter to a signal received from a particular range cell to test whether it matches a given spatio-temporal (angle-Doppler) signature \cite{ward1998space,guerci2014space}.  Though primarily applied in radar, the methods of STAP arise in a multitude of detection and estimation problems: for example, those arising in wireless communications, hyperspectral imaging, and sonar signal processing. 

Common adaptive detectors used in STAP such as the linear filter \cite{reed1974rapid}, the adaptive matched filter \cite[Equation~8]{robey1992cfar}, Kelly's GLRT \cite{kelly1986adaptive}, or the adaptive coherence estimator \cite{kraut1999cfar,kraut2001adaptive,kraut2005adaptive,mcwhorter1996adaptive} depend upon an estimate of the large $p\times p$ interference covariance matrix called the sample covariance matrix, which is formed from $n$ interference-only training samples.
When $n > 2p$, the Reed-Mallett-Brennan rule of thumb states that these detectors can be expected to perform well \cite{reed1974rapid}.
But in STAP, $n$ is not only smaller than $2p$ but often smaller than $p$ due to resolution requirements, the presence of a large number of other targets and target-like scatterers, the fact that the interference statistics are highly non-stationary from range cell to range cell, and systems-level restrictions such as bandwidth \cite{himed1997analyzing}. As a result, in this low-sample-support regime, a multitude of ``robust'' maximum-likelihood covariance estimators have been suggested to replace sample covariance \cite{nitzberg1980application, de2003maximum,ginolhac2014exploiting,li1999computationally, fuhrmann1991application, abramovich1998positive,  barton1997structured,fuhrmann1990estimation,conte1998adaptive,kraay2007physically}. 
But provable properties of these estimators all rely upon the assumption that $n\to\infty$ while $p$ remains fixed, which cannot be the case if $n < 2p$.

In this paper, we resolve this problem by allowing $p$ to go to infinity as well.  To be more precise, we enter the 
 ``large-dimensional asymptotic regime'' of random matrix theory, in which $n$ and $p$ \emph{both} go to infinity and do so in a fixed ratio.  
In this regime, we present a new consistency condition (Definition~2) and accomplish the following provable results: 
\begin{itemize}
\item In Theorem~\ref{thm:S-consistency-1}, we show that a consistent estimator exists under the spiked assumption of Johnstone \cite{johnstone2001distribution}.
\item In Theorem~\ref{thm:nsnr-consistent}, we prove that consistent estimators are detection-theoretic optimal among shrinkage estimators in the formation of filters.
\item In Theorem~\ref{thm:pfa}, we consistently estimate a conditional false-alarm rate of a consistent estimator's filter.
\item In Theorem~\ref{thm:pd}, we characterize the corresponding conditional detection rate. 
\end{itemize}
Notably, our estimates of the conditional detection and false-alarm rates are universal in the sense that depend only on data that is given, even for non-Gaussian interference statistics. 

In Section~\ref{sec:background}, we provide background on STAP detection, as well as material about high-dimensional asymptotics and shrinkage estimators.  In Section~\ref{sec:shrinkage-estimators}, we present our central consistency condition and the optimality result for consistent estimators.  Section~\ref{sec:performance}, we present asymptotic estimates of conditional false-alarm and detection probabilities.
In Section~\ref{sec:simulations}, we present the results of several numerical simulations.
Finally, in Section~\ref{sec:conclusion} we present our conclusions and suggest several directions for future study.

\section{Background} \label{sec:background}
A radar detection system finds targets in a region of interest by transmitting electromagnetic waves toward the region and processing the subsequent reflections, or \emph{echoes}, from objects therein.  An echo is the superposition of the reflections from targets, should they be present, and ``disturbance'' sources, such as one- or multi-bounce reflections from \emph{clutter}, i.e. non-target objects (ground, sea, rain, birds, etc.); electronic emmisions from internal and external sources; electromagnetic interference from man-made sources; and potentially adversarial jamming in the form of electronic noise or false targets. 
Collectively, all non-noise sources of disturbance are referred to as interference.

Modern radar systems have several antennas that transmit a sequence of pulses and passively ``listen'' for echoes in between pulses. If $P$ is the number of pulses in the transmitted sequence and $J$ is the number of antennas, an aggregate $p=PJ$ continuous-time signals are received. Each of these received signals are I/Q demodulated 
and sampled. Prior to detection, the sampled signals are then pre-processed, and the result is a $p \times N$ complex-valued matrix of data, where $N$ denotes the number of samples taken of the return from a single pulse on a single antenna. Each $p$-dimensional column of data corresponds to one of $N$ positions in space along the radial direction of the transmitted electromagnetic waves, called \emph{range cells}. 

Mathematically, the return $\bx\in \mathbb{C}^p$ from a given range cell when the signal is absent is modeled as a mean-zero random vector called the \emph{disturbance vector}.  The $p\times p$ covariance $\bfr$ of the disturbance vector is called the disturbance covariance matrix or interference covariance matrix.  The maximum-entropy distribution for a $p$-dimensional complex random vector with fixed mean $\boldsymbol{\mu}$ and covariance $\bfr$ is the circularly symmetric, complex Gaussian $\mathcal{CN}(\boldsymbol{\mu},\bfr)$.  Because of this fact, and for convenience, the disturbance vector's distribution is often modeled as $\mathcal{CN}(\mathbf{0},\bfr)$.  When the signal is present, 
the ideal return is a multiple of a known spatio-temporal ``steering vector'' $\bs$ by an unknown complex scalar $a$ \cite{wicks2006space}.  In reality, the ideal return is corrupted by the disturbance process, and so the return is modeled as an additive superposition of the scaled steering vector and the disturbance vector.
We therefore wish to test the following compound hypotheses on a return $\bx\in\mathbb{C}^p$:
\begin{equation} \label{eq:prob}
\begin{array}{ll}
\fH_{0}:& \bx\sim \mathcal{CN}(\mathbf{0},\bfr) \\
\fH_{1}:& \bx\sim \mathcal{CN}(a\bs,\bfr),\ a\ne 0.
\end{array}
\end{equation}
In other words, we wish to test the hypothesis $a=0$ versus $a\ne 0$. 

When $\bfr$ is known, a common decision rule for testing
the above hypotheses is the Generalized Likelihood Ratio Test (GLRT) of \cite[Equation~7]{robey1992cfar}, which compares the maximum of the log-likelihood ratio
 over the unknown value of $a$ to a threshold $\tau$:
\begin{equation} \label{eq:AMF}
\frac{\left|\bs'\bfr^{-1}\bx\right|^{2}}{\bs'\bfr^{-1}\bs} \overset{\fH_1}{\underset{\fH_0}{\gtrless}} \tau, 
\end{equation}
where $\bs$ and $\bx$ are column vectors, and $\bs'$ denotes the conjugate transpose of $\bs$. By definition, the \emph{false-alarm rate} of a statistical hypothesis test is the probability of deciding $\fH_1$ when $\fH_0$ is true.  For a threshold test, this is the probability that the test statistic \eqref{eq:AMF} crosses a threshold $\tau$ given that $\fH_0$ is in force.  A threshold test is said to be CFAR (constant false-alarm rate) if its false-alarm rate depends only on $\tau$. This is a highly desirable property as it allows the designer to both set the significance level of the test and ensure that maximum detection probability is obtained for that level.  The GLRT above is known to be a CFAR test of the hypotheses \eqref{eq:prob}.

When $\bfr$ is unknown, it is common to use
the \emph{adaptive matched filter} (AMF) detector
\begin{equation} 
T(\bs,\bfrh, \bx) \coloneqq \frac{\left|\bs'\bfrh^{-1}\bx\right|^{2}}{\bs'\bfrh^{-1}\bs}\overset{\fH_1}{\underset{\fH_0}{\gtrless}} \tau, \label{eq:real-amf}
\end{equation}
where $\bfrh$ is an estimate of the population covariance $\bfr$ obtained from training data that are statistically independent of $\bx$.
We assume there are $n$ such training data 
 $\bx_1, \bx_2,\dots \bx_n$ that are iid distributed as $\fH_0$.
For example, 
the training data can be snapshots from target-free range bins near the cell under test.
We will often write such training samples as a $p \times n$ matrix:
\def\tmp{
\begin{bmatrix}
\vertbar & \vertbar & ~ & \vertbar \\
\bx_1 & \bx_2 & \cdots & \bx_n \\
\vertbar & \vertbar & ~ & \vertbar
\end{bmatrix}.
}
\[
\bX_n \coloneqq 
\stackMath\def\stackalignment{r}%
  \stackon%
    {\ensuremath{\normalsize p} \left\{\tmp\right.}%
    {\overbrace{\phantom{\smash{\tmp\mkern -36mu}}}^{\scriptstyle\textstyle \ensuremath{n}}\mkern 25mu}%
\]

Under the assumed Gaussian-distributed hypotheses, the detection probability of the test \eqref{eq:AMF} is monotonic in the quantity $|a|^2\bs'\bfr^{-1}\bs$, known as
 the \emph{signal-to-interference-plus-noise-ratio} (SINR) of the associated filter. 
By contrast the conditional detection probability of the AMF, given the training data, is 
monotonic in
the \emph{effective} 
SINR of the filter, introduced in \cite{reed1974rapid}:
\begin{equation} \label{eq:rho-squared}
\nu^{2}(\bs, \bfrh,\bfr):=|a|^{2}\frac{\left(\bs'\bfrh^{-1}\bs\right)^{2}}{\bs'\bfrh^{-1}\bfr\bfrh^{-1}\bs}.
\end{equation}
To see this, let us consider $p_{\text{fa}}$ and $p_\text{d}$, the conditional false-alarm and detection rates of the test in \eqref{eq:real-amf} given the training data. Since $\bx$ is Gaussian, $T$ is, conditioned on $\bX_n$, a chi-square random variable scaled by $\xi = \bs'\bfrh^{-1}\bfr\bfrh^{-1}\bs/\bs'\bfrh^{-1}\bs$ under $\fH_0$, and thus $p_{\text{fa}} = e^{-\tau/\xi}$.  Under $\fH_1$, $T$ is, conditioned on $\bX_n$, a scaling of a noncentral chi-square distribution with noncentrality parameter $|a|^2\bs'\bfrh^{-1}\bs$, where the scaling factor is again $\xi$.  Thus, $p_{\text{d}} = Q(\tau/\xi, \nu^2)$, where
\[
Q(\alpha, \beta) \coloneqq \int_\alpha^\infty e^{-z-\beta}I_0(2\sqrt{\beta z})\, dz
\]
and $I_\mu$ denotes the modified Bessel function of the first kind.  The form of $Q$ follows from \cite[p.~132]{johnson1970continuous}. As a result, 
\[
p_{\text{d}} = Q(-\log p_{\text{fa}}, \nu^2),
\]
which is monotonic in $\nu^2$, as claimed.

The conditional detection probability $p_\text{d}$ is also 
monotonic in the
\emph{normalized signal-to-interference-plus-noise ratio} (NSINR) \cite{reed1974rapid}:
\begin{equation} \label{eq:esnr}
\eta(\bs,\bfrh, \bfr) \coloneqq \frac{(\bs'\bfrh^{-1}\bs)^2}{(\bs'\bfr^{-1}\bs)(\bs'\bfrh^{-1}\bfr\bfrh^{-1}\bs)}. 
\end{equation}
By the Cauchy-Schwarz inequality, this quantity lies between 0 and 1.  Converting NSINR to decibels and taking the absolute value, one obtains the Reed-Mallet-Brennan (RMB) loss; hence, for a given covariance estimate $\bfrh$, a higher NSINR yields a lower RMB loss, and vice versa.

The random variable $p_\text{d}$, as a function of $\bfrh$, or to monotonic equivalence, $\eta$ as a function of $\bfrh$, is an example of a reward function.  In finite-sample theory, one is given a sample of a fixed size and the goal is to find an estimator so that reward is maximized.  However, this is often intractable.  Instead, one often performs an \emph{asymptotic} analysis. This is accomplished by several steps: (a) the problem is embedded in a sequence of estimation problems of increasing sample size $n$, (b) an estimator is proposed for each problem, and (c) the limiting form of the reward is derived as $n\to\infty$.

When applied to covariance estimation, one embeds  
 the described covariance estimation problem into a sequence of covariance estimation problems indexed by the number of samples $n$, and 
asymptotic values of $p_{\text{fa}}$ and $p_\text{d}$ are computed for an estimator $\bfrh_n$.   
 Assuming $p$ is fixed and $n > p$, the classical Sample Covariance Matrix (SCM) can be easily analyzed in the asymptotic limit. This estimator is defined by
\begin{equation*} 
\mathbf{S}_n = \frac{1}{n}\sum_{i=1}^n \bx_i \bx_i' = \frac{1}{n} \bX_n \bX_n',
\end{equation*}
where the vector products in the summation above are rank-one outer product matrices.
 In the Gaussian setting, the SCM is the
maximum likelihood estimator of the population covariance $\bfr$ and is a consistent estimator. It follows that $\mathbb{E}\eta(\bs,\bS_n,\bfr)\to 1$ as $n\to\infty$ for $\left\Vert \bs \right\Vert > 0$, so $\bS_n$ asymptotically maximizes detection performance.


In applications such as STAP, not only $n$ but also $p$ is large, and $n$ is not much larger than $p$.  This makes the fixed-dimensional asymptotics inapplicable.
Instead, we consider the \emph{high-dimensional} case, where $n$ and $p=p_n$ \emph{both} go to infinity.  The goal, as before, is to optimize the detection rate, but unlike before,  the dimension of the $n^\text{th}$ problem increases with $n$.  
Although other reward functions have been studied in the high-dimensional limit, as in \cite{donoho2018optimal}, to our knowledge the detection rate of an adaptive matched filter has not.



The $n^\text{th}$ problem in question is discriminating between the following hypotheses: 
\begin{equation} \label{eq:projected-hyp-test}
\begin{array}{ll}
\fH_{0}^n:& \bx\sim \mathcal{CN}(0,\bfr_n) \\
\fH_{1}^n:& \bx \sim \mathcal{CN}(a\bs_n,\bfr_n),\ a\ne 0,
\end{array}
\end{equation}
where for each $n$, $\bfr_n$ is a $p_n\times p_n$ Hermitian positive-definite matrices with smallest eigenvalue 1 and $M \coloneqq \left\Vert \bfr_n \right\Vert$ is fixed and finite, and $\bs_n$ is uniformly distributed on the unit sphere in $\mathbb{C}^{p_n}$.
The latter assumption is simply a modeling assumption which will enable us to say whether an estimator performs well for ``most'' steering vectors.  Note that there is no loss in generality in assuming $\bs_n$ is a unit-norm vector.
For the $n^\text{th}$ problem, we will assume the availability of a $p_n \times n$ matrix of $\fH_0^n$-distributed auxiliary training data $\bX_n$ independent of $\bs_n$. We further make the following modeling assumption throughout this paper:
\begin{itemize}
\item $[\textsc{Asy}(\myc)]$  The number of samples $n$ and the number of dimensions $p_n$ in each sample follow the proportional-growth limit $p_n/n \to \myc \in (0,1)\cup (1,\infty)$ as $n\to \infty$. 
\end{itemize}
\noindent 
The assumption [$\textsc{Asy}(\myc)$] appears in \cite{ledoit2011eigenvectors}.
A consequence of $[\textsc{Asy}(\myc)$] is that throughout this paper, when we write $n\to\infty$ it will be assumed that $p_n\to\infty$ as well.  The assumption that $\fH_0^n$ and $\fH_1^n$ are Gaussian is convenient but certainly not necessary for what follows: by the Berry-Esseen theorem \cite{berry1941accuracy,esseen1942liapunoff}, a properly normalized matched filter applied to $\bx$ is distributed asymptotically the same as in the Gaussian case, provided a mild decay condition is met.  In particular, the asymptotic conditional detection rate of $T_n = T(\bs_n,\bfrh_n,\bx)$ is completely determined by NSINR $\eta_n = \eta(\bs_n,\bfrh_n,\bfr_n)$ even in the non-Gaussian case. As a result, in general a sensible choice of $\bfrh_n$ is one that optimizes $\eta_n$. 
A particular kind of covariance estimator called a \emph{shrinkage estimator} has been popular since at least the time of C. Stein \cite{stein1975estimation,stein1986lectures}. This term certainly includes diagonal loading estimators, which occur in radar \cite{fuhrmann1988existence}, mathematical finance \cite{ledoit2004well}, Tikhonov regression \cite{tikhonov1943stability}, and many other areas.  It can also more generally mean any estimator that shares the eigenspace decomposition of the sample covariance matrix \cite{ledoit2018analytical,ledoit2020analytical,donoho2018optimal}.  In this paper, we take the latter definition, adding explicitly a condition (condition (ii) below) that is virtually always satisfied in practice.  The definition follows:
\begin{definition} \label{def:shrinkage}
Let $\mathcal{P}_{p_n}$ be the cone of $p_n\times p_n$ Hermitian positive-definite matrices.  Let $\bS_n = n^{-1}\bX_n \bX_n'$.  We say that $\bfrh_n:\mathbb{C}^{p_n\times n} \to \mathcal{P}_{p_n}$ is a \emph{shrinkage estimator} if
\begin{itemize}
\item[(i)] $\bfrh_n$ is of the form $\bU_n \bD_n \bU_n'$ where $\bU_n$ is a random element of the family of matrices such that $\bU_n' \bS_n \bU_n$ is diagonal, and $\bD_n$ is a positive-definite, diagonal random matrix.
\item[(ii)] the random variables $\limsup_n \left\Vert \bfrh_n \right\Vert$ and $\sup_n\left\Vert \bfrh_n^{-1} \right\Vert$ are almost surely bounded.
\end{itemize}
\end{definition}  

One of the simplest examples of a shrinkage estimator is a positive linear combination of $\bS_n$ and the $p_n\times p_n$ identity matrix $\mathbf{I}_{p_n}$: just take $\mathbf{D}_n(\boldsymbol{\Lambda})=\alpha\boldsymbol{\Lambda} + \beta \mathbf{I}_{p_n} $.  The $\limsup$ condition in Definition~\ref{def:shrinkage}(ii) holds because if $\bZ_n=\bfr_n^{-1/2} \bX_n$,  $\left\Vert\alpha\bS_n+\beta\mathbf{I}_{p_n} \right\Vert \le \alpha M \left\Vert \bZ_n \bZ_n' /n\right\Vert + \beta$, and $\left\Vert\bZ_n\bZ_n'/n\right\Vert$ converges almost surely under a fourth-moment condition to $(1+\sqrt{\myc})^2$ \cite{bai2008limit}.  
Many structure-constrained maximum likelihood estimators 
\cite{steiner2000fast,anderson1963asymptotic,wax1985detection} provide further examples. The name ``shrinkage'' is motivated by the fact that many shrinkage estimators reduce higher eigenvalues of $\bS_n$ and possibly increase lower ones (see \cite{ledoit2004well}), thus ``shrinking'' the spectrum of $\bS_n$.  The result is an estimator that shares the eigenspace decomposition of sample covariance but improves its condition number. 

\section{Optimal Shrinkage Estimators} \label{sec:shrinkage-estimators}


A central question in STAP is how to choose a shrinkage estimator $\bfrh_n$ so that $\eta(\bs_n, \bfrh_n, \bfr_n)$ is as large as possible.
The answer depends on $\bs_n$ in a complicated manner; however, this question turns out to be tractable in high dimensions.
Indeed, in the following lemma we show that $\eta(\bs_n,\bfrh_n,\bfr_n)$ is asymptotically independent of $\bs_n$.
\begin{lem} \label{lem:trace-approx-eta}
Let $\bfrh_n $ be a sequence of shrinkage estimators.  Then we have
\begin{equation} \label{eq:trace-approx-eta}
\left \vert \eta(\bs_n, \bfrh_n, \bfr_n) - \tilde{\eta}(\bfrh_n,\bfr_n) \right \vert \overset{a.s.}{\to} 0,
\end{equation}
as $n\to\infty$, where
\begin{equation} \label{eq:eta-tilde}
\tilde{\eta}(\bfrh,\bfr) :=
\frac{\tr(\bfrh^{-1})^2}{\tr(\bfr^{-1})\tr(\bfrh^{-2}\bfr)}.
\end{equation}
\end{lem}
\begin{proof}
See Appendix~\ref{sec:trace}.
\end{proof}
	

Motivated by the above lemma, we use $\tilde{\eta}$ as a proxy for $\eta(\bs_n, \, \cdot\, )$ and seek to find the optimal generalized shrinkage estimator with respect to the former.   Let $\bfrh_n  = \bU_n \bD_n \bU'_n$, where $\bU_n'\bS_n\bU_n$ is diagonal and $\bD_n$ is an arbitrary diagonal matrix. Let $\bU_n = [\bu_{n,1},\bu_{n,2},\dots \bu_{n,p_n}]$.  Then we have
 \begin{align} 
\tilde{\eta}(\bfrh_n,\bfr_n) & = \frac{\tr(\bU_n \bD_n^{-1} \bU_n')^2}{\tr(\bfr_n^{-1})\tr(\bU_n\bD_n^{-2}\bU_n'\bfr_n)} \nonumber \\
& = \frac{\tr(\bD_n^{-1})^2}{\tr(\bfr_n^{-1})\tr(\bD_n^{-2}\bU_n' \bfr_n \bU_n)},\label{eq:eta-tilde-D}
 \end{align}
 where we have used the cyclic-permutation property of trace. 
 We show in Appendix~\ref{sec:eta-tilde-D} that $\tilde{\eta}(\bfrh_n,\bfr_n)$ is maximized when 
 \begin{equation} \label{eq:oracle-eivals}
(\bD_n)_{ii} = d^*_{n,i} :=  \buni' \bfr_n \buni.
 \end{equation}
 Thus, in terms of maximizing $\tilde{\eta}$ the shrinkage estimator $\bfrh_n$ is at most as good as $\bfrs_n \coloneqq \bU_n \bD^*_n  \bU_n'$,
 where
 \[
\bD^*_n = \mathrm{diag} (d^*_{n,1},d^*_{n,2},\dots d^*_{n,p_n}).
 \]
  We call $\bfrs_n$ 
 a \emph{shrinkage oracle}, or just \emph{oracle}.

 Let
 \[
	L_{p_n}(\bfrh_n, \bfr_n) = \frac{1}{p_n}\left\Vert \bfrh_n - \bfr_n \right\Vert_{\text{F}}^2,
	\]
where $\left\Vert \cdot \right\Vert_{\text{F}}$ denotes the Frobenius norm.
 The following defines estimators that are ``just as good'' as an oracle.
 \begin{definition} \label{def:s-consistency}
          Let $\bfrh_n:\mathbb{C}^{p_n\times n}\to\mathcal{P}_{p_n}$ and let $\bfrs_n$ be a sequence of shrinkage oracles.  Then we say that
	$\bfrh_n$ is \emph{oracle consistent} if $\bfrh_n$ is a shrinkage estimator and
	\begin{equation} \label{eq:frob-consistency}
		L_{p_n}(\bfrh_n,\bfrs_n) \toprob 0
		\end{equation}
	as $n\to\infty$.
	\end{definition}
 
In Appendix~\ref{sec:lw-consistency}, we give a constructive proof of the following theorem using the work of \cite{ledoit2011eigenvectors,ledoit2017direct}.  We note that the ``spiked'' assumption is reasonable in STAP since interference covariances are often modeled as ``low-rank plus noise'' and their smallest eigenvalues are often assumed known \cite{steiner2000fast}.

\begin{thm}\label{thm:S-consistency-1}
Assume the ``spiked'' model of Johnstone \cite{johnstone2001distribution}: $R_n = \mathrm{diag}(\tau_{n,1}, \tau_{n,2}, \dots, \tau_{n,p_n})$ and all but the largest $r$ eigenvalues are 1, where $r$ is independent of $n$.  Further, suppose the largest $r$ eigenvalues are fixed and independent of $n$.  Then there exists an oracle-consistent estimator.
 \end{thm}
 
 \begin{rem}
We have assumed that the columns of $
\bX_n$ are Gaussian-distributed for ease of exposition, but by \cite{ledoit2011eigenvectors,ledoit2017direct}, the result is much more distribution-free.  Indeed, it is only necessary to assume that the components of $\bfr_n^{-1/2}\bX_n$ have finite absolute central sixteenth moment bounded by a constant independent of $n$ and $p_n$ \cite{ledoit2017direct}.  This moment condition, in turn, can likely be relaxed to the much more lenient finite fourth moment assumption that is common in random matrix theory. 
\end{rem}

Any oracle consistent estimator is optimal in the following sense:

 \begin{thm} \label{thm:nsnr-consistent}
Let $\bfrh_n$ be oracle consistent.
Then
\begin{equation*}
\left \vert \eta(\bs_n, \bfrh_n, \bfr_n) - \eta(\bs_n, \bfrs_n, \bfr_n) \right \vert \toprob 0,
\end{equation*}
as $n\to\infty$.
\end{thm}
\begin{proof} See Appendix~\ref{sec:ab-lemma}.

\end{proof}


In the next Section we investigate conditional false-alarm
and detection probabilities of the filter formed from an oracle consistent estimator $\bfrh_n$.

\section{Performance Analysis of Oracle Consistent Estimators} \label{sec:performance}
In this Section, we derive analytical asymptotically consistent performance estimates for the detector $T_n = T(\bs_n,\bfrh_n, \bx)$ formed
from an oracle consistent estimator $\bfrh_n$.  As in Section~II we assume $\bx$ is Gaussian for convenience, but this is certainly not necessary: the selfsame results hold regardless of distribution.
%

In the rest of this section, the key lemma will be the following.
\begin{lem} \label{lem:limit-of-xi}
If $\bfrh_n$ is oracle consistent, then
\[
	\xi_n := \xi(\bs_n, \bfrh_n, \bfr_n) \toprob 1
	\]
as $n\to\infty$, where
\[
	\xi(\bs, \bfrh, \bfr)\coloneqq \frac{\bs'\bfrh^{-1}\bfr\bfrh^{-1}\bs}{\bs'\bfrh^{-1}\bs}.
	\]
\end{lem}
\begin{proof}
See Appendix~\ref{sec:limit-of-xi}.
\end{proof}

\subsection{False-Alarm Rate}
The conditional false-alarm rate of $T_n$ given the training data $\bX_n$ and a random steering vector $\bs_n$ is the random variable given by
\begin{align*}
	& p_{\text{fa}}^n(\tau)  \coloneqq
	 \Pr\left[  T_n >\tau \mid \fH_0^n, \bs_n, \bX_n \right].
	\end{align*}
	
In the following theorem, we present an asymptotically consistent estimate of this rate that is
independent of $\bX_n$, $\bs_n$, and the unknown sequence of matrices $\{ \bfr_n \}_{n = 1}^\infty$.   This means the detector has the extremely useful \emph{CFAR property}, like its cousin that is based on sample covariance \cite{robey1992cfar}.  However, unlike its cousin this detector's conditional false-alarm rate converges to a limit that is both \emph{non-random} and \emph{closed-form}.  This means that this test is asymptotically as good as the GLRT of \cite{robey1992cfar} in the sense that the limiting false-alarm rate can be set exactly using the threshold alone---a highly desirable property from a statistical standpoint.
\begin{thm} \label{thm:pfa}
	If $\bfrh_n$ is oracle consistent, then
	\[
		p_{\textup{fa}}^n(\tau) \toprob e^{-\tau},
	\]
	as $n\to\infty$. 
	\end{thm}
\begin{proof}
	Fix $\tau \in \mathbb{R}$. The statistic $T_n$ is a scaled complex chi-square random variable, so
	\begin{equation} \label{eq:fa-rate-nonasymptotic}
	p_{\text{fa}}^n(\tau) = \exp(-\tau/\xi_n).
	\end{equation}
The function $h(x) = e^{-\tau/x}$ is continuous at every point $x \in \mathbb{R}_+$. By Lemma~\ref{lem:limit-of-xi},
	we have  
	\[
		\xi_n \toprob 1
	\]
	as $n\to\infty$.
	Continuous functions preserve convergence in probability \cite{mann1943stochastic}, hence
	\[
		h(\xi_n) \toprob h(1)
	\]
	as $n\to\infty$
	That is, 
	\[
		\exp(-\tau/\xi_n) \toprob e^{-\tau}
	\]
	as $n\to\infty$.
	By \eqref{eq:fa-rate-nonasymptotic}, this is the desired result.
	\end{proof}

In the next section we will obtain a similar result relevant to the detection rate.

\subsection{Detection Rate}

%
In this section, we show how to estimate the conditional detection rate of $T_n$ given $\bs_n$ and $\bX_n$:
\begin{align*}
	& p_{\text{d}}^n(\tau) 
	\coloneqq \Pr\left[  T_n >\tau \mid \fH_1^n, \bs_n, \bX_n \right].
	\end{align*}
It follows from the distribution of $T_n$ that $p_{\text{d}}^n(\tau) = Q(\tau/\xi_n, \nu_n^2)$.
Ideally, then, one thing we would like to know is $\nu_n$.  However, complications such as unknown radar cross section make it necessary to estimate this quantity.  In Lemma~\ref{thm:sinr} below, we provide just such an estimate.  This estimate uses the only information we have about $\nu_n$---namely, the test datum $\bx$.
 

Before we state the lemma, let us introduce a bit of terminology. For any positive definite matrix $\mathbf{P}$ and any properly sized column vectors $\bs$ and $\bx$, let 
\begin{equation} \label{eq:noncentrality-approx}
\hat{\nu}(\bs, \mathbf{P}, \bx)=\frac{|\bs'\mathbf{P}^{-1}\bx|}{(\bs'\mathbf{P}^{-1}\bs)^{1/2}}.
\end{equation}
Further, if $X_n$ and $Y_n$ are random variables
we will say that $X_n$ is  \emph{asymptotically less than or equal to} $Y_n$, denoted $X_n \lesssim Y_n$, iff 
$\max\{X_n-Y_n,0\} \toprob 0$ as $n\to\infty$.  
``Asymptotically greater than or equal to'' is defined similarly.
The following lemma states that if $\bfrh_n$ is oracle consistent, $\hat{\nu}_n(\bx)\coloneqq \hat{\nu}(\bs_n, \bfrh_n, \bx)$ is  a reasonable estimator of $\nu_n$.
\begin{lem} \label{thm:sinr}
Suppose $\bfrh_n$ is oracle consistent. Then 
\[
\Pr\left[\left| \hat{\nu}_n(\bx) -\nu_n \right| \ge t \mid \fH_1^n, \bs_n, \bX_n \right] \lesssim e^{-t^2}
\]
as $n\to\infty$.
\end{lem}
\begin{proof} See Appendix~\ref{sec:pf-thm-sinr}.
\end{proof}

We note that the estimate above is essentially the same as the one in \cite{robinson2019optimal}, but $\hat{\nu}_n(\bx)$ is a significantly tighter estimator in practice than the one in that paper.

Roughly speaking, if we apply $Q(\tau, \cdot)$ to this result, we get the following characterization of the conditional detection probability $p_\text{d}^n(\tau)$ in terms of confidence intervals.

\begin{thm} \label{thm:pd}
Suppose $\bfrh_n$ is oracle consistent.  Let the ``confidence'' be $q\in [0,1)$, let $\tau \ge 0$, and let
\[
\hat{\nu}_{n\pm}(\bx,\myq) = \max\left\{0, \hat{\nu}_n(\bx)\pm\sqrt{\log\frac{1}{1-\myq}}\right\}.
\]
Then the probability that $p_{\textup{d}}^n(\tau)$ lies between $Q(\tau, \hat{\nu}_{n-}(\bx,\myq)^2)$ and $Q(\tau,\hat{\nu}_{n+} (\bx,\myq)^2)$ given  $\fH_1^n,\bs_n,$ and $\bX_n$ is asymptotically greater than or equal to $q$, as $n\to\infty$.
\end{thm}
\begin{proof} See Appendix~\ref{sec:pd}. 
\end{proof}

We finally note that all of the estimates contained in this section are \emph{bona fide} estimates, in the sense that they depend only on known quantities.

%
 
 \section{Simulations} \label{sec:simulations}
 
In this section we compare several popular covariance estimators to the estimator described in Appendix~C, which we call the Ledoit-Wolf Direct (LWD) estimator. 

\subsection{Alternative Estimators}

Below we list several popular covariance estimators arising in STAP.  To define all estimators, fix a sample $\bX_n$ of size $n$ whose columns have covariance $\bfr$ and let $\lambda_1\le  \dots\le \lambda_p$ be the eigenvalues of sample covariance $\bS_n=\bX_n\bX_n'/n$.

\subsubsection{DGJ} \cite{donoho2018optimal} In the PCA literature, several recent results have shown that sample eigenvalues and eigenvectors in the spiked model are biased in a predictable deterministic way from their population counterparts in the large-dimensional limit.  Donoho, Gavish, and Johnstone have used these biasing formulae to propose a shrinkage estimator that is asymptotically as close to $R_n$ as the oracle.  Such an estimator is an example of an oracle-consistent estimator, but it is only defined for $n\ge p$ at the moment and may require very large $n$ and $p_n$ to converge.

\subsubsection{Anderson-42}\cite{anderson1963asymptotic} This estimator assumes a spiked structure and that the rank $r$ is known.  In this case, an estimate of the smallest population eigenvalue is computed:
\[
\hat{\sigma}^2 = \frac{1}{p-r} \sum_{i=1}^{p-r} \lambda_i.
\]
The estimator is obtained by thresholding $\bS_n$ from below by $\hat{\sigma}^2$.  The only question is how to approximate $r$.  A choice made in the literature for the data set described in Subsection~\ref{subsec:comparison-to-oracle} is $r=42$ \cite{kang2014rank}.

\subsubsection{FML} \cite{steiner2000fast} Steiner and Gerlach's Fast Maximum Likelihood is the maximum-likelihood estimator subject to the constraint that the smallest eigenvalue is known.  In Subsection~\ref{subsec:comparison-to-oracle} and many places, this eigenvalue is noralized to $\sigma^2=1$.  The result is obtained by thresholding $\bS_n$ from below by $\sigma^2$.


\subsubsection{LW diagonal loading} For a given covariance $\bfr$, there is an oracle scaled convex 
combination of 
sample covariance and the identity, as described in Ledoit and Wolf's \cite{ledoit2004well}.
The \emph{bona fide} estimator described in that paper is an approximation
to the oracle linear combination that converges (in a sense in the quartic mean)
as $n,p_n \to\infty$.  
This estimator can be described as the Frobenius-norm optimal diagonal loading 
estimator.

\subsection{Comparison of Estimators' NSINR} \label{subsec:comparison-to-oracle}

In this subsection, we take the population covariance $\bfr$ to be the ideal covariance from range bin 1 of the KASSPER I \cite{guerci2006knowledge} 
data set (Knowledge-Aided Sensor Signal Processing and Expert Reasoning).  KASSPER I is a high-fidelity, physics-based simulation of
radar data collected by a multichannel array over multiple pulses. 
This data set is ideal for testing
the kind of detector under consideration
since truth covariances are included in the data set.  In these figures, $p=JP=352$ is fixed, where the number of antennas $J=11$ and the number of pulses $P=32$.  This covariance roughly conforms to the spiked model, but also stretches its limits of the spiked assumption, with over 42 population eigenvalues exceeding 1.


In Figures~\ref{fig:medians-raw} and \ref{fig:NSNR-thresh-unzoomed}, we plot the median NSINRs of the estimators above computed from 100 trials for each value of $n$ in the range $40,60,80\dots 500$. 
Here, the training matrices are Laplace distributed with covariance $\bfr$, meaning they are matrices of white Laplace noise colored by $\bfr$.
%
It is easy to see 
that the median performance of LWD
is almost indistinguishable from the median performance of the shrinkage oracle and that FML trails closely behind.  Where defined (i.e., when $n \ge 352$), DGJ also rivals LWD in terms of NSINR.


We note that LWD is only predicted to perform well for $\gamma\ne 1$.  In practice, inconsistencies and suboptimality could be encountered for $\gamma$ near to 1.  This was not an issue in our experiments, but more study is needed to understand whether it could be an issue for other covariances and values of $n$.

%

\subsection{Performance Predictions}

The performance predictions of Theorem~\ref{thm:pfa} and Theorem~\ref{thm:pd} hinge crucially on the convergence of $\xi_n$ to 1, as defined in Lemma~\ref{lem:limit-of-xi}.  Since it is difficult to show the dependence of $p_{\text{fa}}$ and $p_{\text{d}}$ on $n$ and $p$ as $n,p\to\infty$ by Monte Carlo simulation, we show instead the behavior of $\xi_n$ for the LWD, FML, and DGJ estimator for $\bfr$ equal to the same KASSPER covariance of the last subsection.  In Figure~\ref{fig:pfa}, we show the dependence of 100 trials of $\xi_n$ on $n$, plotting 10-90 percentiles.  As can be seen, the median values of $\xi_n$ for LWD are closer to 1 than FML's or DGJ's for all $n$ displayed.  That they are closer to 1 than FML's is to be expected: FML is not oracle-consistent.  That they are closer to 1 than DGJ's seems to indicate LWD converges to an oracle more quickly than DGJ does. 

\begin{figure}
	\begin{center}
		\includegraphics[trim={2.25cm 7cm 2.25cm 7cm},clip,scale=.475]{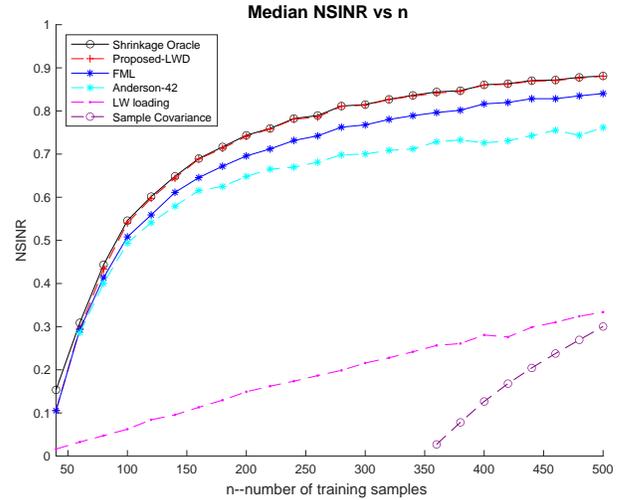}
		\end{center}
		\caption{\label{fig:medians-raw}
		A plot that shows nearly identical performance of LWD and the shrinkage oracle.
		}
\end{figure}

\begin{figure}
	\begin{center}
		\includegraphics[trim={2.25cm 7cm 2.25cm 7cm},scale=.475]{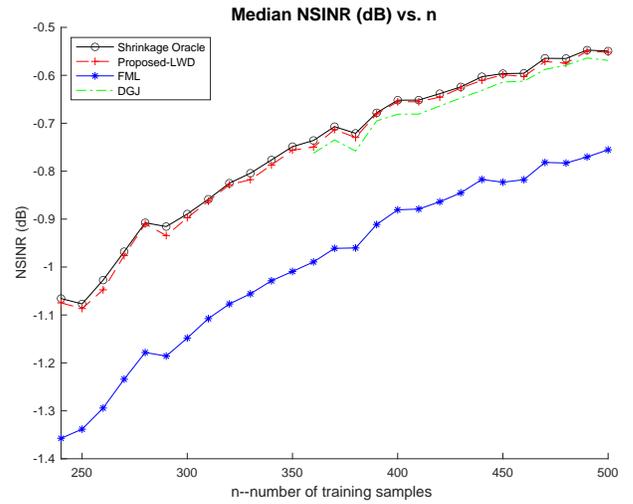}
		\end{center}
\caption{\label{fig:NSNR-thresh-unzoomed} A close-up of Figure~\ref{fig:medians-raw}, in decibels, including DGJ.}
\end{figure}


\begin{figure}
	\begin{center}
		\includegraphics[trim={2.25cm 7cm 2.25cm 7cm},scale=.475]{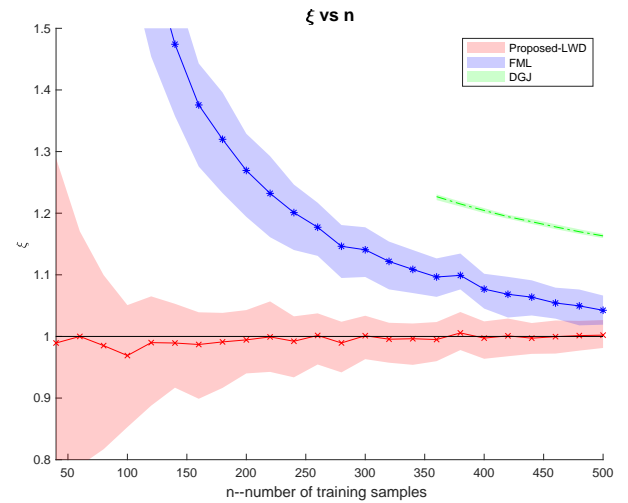}
		\end{center}
\caption{\label{fig:pfa} A plot showing $\xi$ for LWD, FML, and DGJ showing LWD closer to 1 on average for all $n$ displayed.}
\end{figure}

\section{Conclusion} \label{sec:conclusion}
In this paper we have proposed a new oracle-consistency condition (Definition~\ref{def:s-consistency}) for covariance estimators which enables the development of closed-form asymptotically consistent performance estimates for the corresponding adaptive matched filter (Theorems~\ref{thm:pfa} and \ref{thm:pd}) that depend only on given data, even for non-Gaussian interference statistics.  We have shown oracle consistent shrinkage estimators exist (Theorem~\ref{thm:S-consistency-1}) in some special situations and are detection-theoretic optimal among shrinkage estimators (Theorem~\ref{thm:nsnr-consistent}).  Further, we have shown in Section~\ref{sec:simulations} that the given example performs as expected in simulation.


Future work may include relaxing the spiked assumption, relaxing the assumption that the smallest population eigenvalue is known, dealing with the case of $\gamma\approx 1$, and pursuing the rate of convergence to the shrinkage oracle.

\appendices
\section{Proof of Lemma~\ref{lem:trace-approx-eta}} \label{sec:trace}

Throughout the appendices, convergence of random variables means convergence as $n\to\infty$.

It follows immediately from the definition of uniform convergence that if $x_n$ and $y_n\in \mathbb{R}^d$ satisfy $x_n - y_n \to 0$ and $f: \mathbb{R}^d \to \mathbb{R}$ is uniformly continuous, then $f(x_n)-f(y_n) \to 0$.

The result is a consequence of the following simple lemmas with $f(x,y,z)=x^2/(zy)$.

\begin{lem} \label{lem:quotients}
Suppose $X_n$ and $Y_n$ are random vectors in $\mathbb{R}^d$ and that $f$ is uniformly continuous on the essential ranges of $X_n$ and $Y_n$   \cite{rudin2006real}.  Suppose $X_n - Y_n$ goes to zero, either in probability or almost surely. Then $f(X_n)-f(Y_n)$ goes to zero in probability or almost surely, respectively.
\end{lem}

\begin{lem} \label{lem:trace}
Let $\bfrh_n$ be a sequence of random matrices for which the random variables $\limsup_n \left\Vert \bfrh_n \right\Vert$ and $\sup_n \left\Vert \bfrh_n^{-1}\right\Vert$ are almost surely bounded and $\bfrh_n$ is independent of $\bs_n$. 
 Let
\[
\underline{X}_n = \left(\bs_n'\bfrh^{-1}_n \bs_n ,  \bs_n'\bfrh^{-1}_n\bfr_n\bfrh_n^{-1} \bs_n, \bs_n'\bfr^{-1}_n \bs_n\right)
\]
and
\[
\tilde{\underline{X}}_n = p_n^{-1} \left(\tr(\bfrh_n^{-1}), \tr(\bfrh_n^{-2}\bfr_n), \tr(\bfr_n^{-1}) \right).
\]
Then the essential ranges of $X_n$ and $\tilde{X}_n$ do not include zero  
and
\[
\underline{X}_n - \tilde{\underline{X}}_n \toas 0.
\]
\end{lem}

%

To prove Lemma~\ref{lem:trace}, we will need a couple of supporting results.  The first concerns approximating a trace using a quadratic form.
\begin{lem} \label{lem:trace-approx}
Let $\bfy$ be a random complex column vector that is uniformly distributed on the sphere in $\mathbb{C}^p$.  Let also $\bfA$ be a complex $p\times p$ matrix.  Then there exists a constant $c > 0$ independent of $p$ and $\bfA$ such that for all $\epsilon > 0$ we have
\[
\Pr\left[ \left|\bfy' \bfA \bfy - \frac{1}{p}\tr \bfA \right| \ge \epsilon \right] \le \exp\left(-cp\epsilon^2/\left\Vert \bfA \right\Vert^2 \right).
\]
\end{lem}
\begin{proof}
Let $\bu$ be a random vector uniformly distributed on the sphere of radius $\sqrt{p}$ in $\mathbb{C}^p$.  It follows immediately from \cite[Theorem~5.1.4]{vershynin2018high} that if $f$ is Lipschitz  on this sphere and $t>0$, there exists $c' >0$ independent of $f$ and $t$ such that
\[
\Pr\left[ \left| f(\bu) - \mathbb{E}f(\bu) \right| \ge t \right] \le \exp(-c' t^2/L_f^2),
\]
where $L_f$ is a Lipschitz constant of $f$.  It is well-known that $\mathbb{E}[\bu\bu'] = \mathbf{I}$, the $p\times p$ identity matrix.
Thus, $$\mathbb{E} \bu'\bfA \bu = \mathbb{E} \tr(\bfA\bu\bu') = \tr\left(\bfA \mathbb{E}[\bu\bu'] \right)= \tr \bfA.$$
Further, since the gradient of $f(\bu)=\bu'\bfA\bu$ is $2\bfA \bu$, a Lipschitz constant of $f$ is easily seen to be $2\left\Vert \bfA\right\Vert \left\Vert \bu \right\Vert$, which is equal to $2\sqrt{p} \left\Vert \bfA\right\Vert $ on the sphere in question.  Taking $\bfy = \bu/\sqrt{p}$, then, we get
\begin{align*}
& \Pr\left[ \left|\bfy' \bfA \bfy - \frac{1}{p}\tr \bfA \right| \ge \epsilon \right] \\
& \Pr\left[ \left|\bu' \bfA \bu - \tr \bfA \right| \ge p \epsilon \right] \\
& = \Pr\left[\left| f(\bu) - \mathbb{E}f(\bu) \right| \ge p\epsilon \right] \\
& \le \exp\left(-c'(p\epsilon)^2/(2\sqrt{p}\left\Vert\bfA\right\Vert)^2\right).
\end{align*}
The result follows by taking $c= c'/4$.
%
\end{proof}


The second preliminary lemma converts the approximation in Lemma~\ref{lem:trace-approx} into almost sure convergence.
\begin{lem}
Let $\mathbf{A}_n \in \mathbb{C}^{p_n \times p_n}$ be a positive-definite random matrix such that the random variable $\sup_n \left\Vert \mathbf{A}_n \right\Vert$ is almost surely bounded. 
Then
\[
\left \vert \bs_n' \mathbf{A}_n \bs_n - \frac{1}{p_n} \tr(\mathbf{A}_n) \right \vert \toas 0.
\]
\end{lem}
\begin{proof}
Let $\epsilon > 0$. By the Borel-Cantelli lemma (see, for example, \cite[Theorem~10.10]{folland1999real}), it suffices to prove that
\[
 \Pr\left[\left|\mathbf{s}'_{n}\mathbf{A}_{n}\mathbf{s}_{n}-\frac{1}{p_{n}}\tr(\mathbf{A}_{n})\right|\ge\epsilon\right] 
\]
goes to zero at a rate of $O(n^{-1-\delta})$ for some $\delta > 0$.  
By the definition of conditional probabiilty the above probability is equal to
\begin{equation}
\mathbb{E} \left ( \Pr\left[\left. \left|\bs'_{n}\mathbf{A}_{n}\bs_{n}-\frac{1}{p_{n}}\tr(\mathbf{A}_{n})\right|\ge\epsilon\ \right| \mathbf{A}_{n}\right] \right ). \label{eq:trace-approx-01}
\end{equation}
Now, since
$\sup_n \left\Vert \bfA_n \right\Vert$ is bounded above by some constant $D$ with probability 1, it follows from Lemma \ref{lem:trace-approx} that with probability 1 we have
\begin{align*}
&  \Pr\left[\left. \left|\bs'_{n}\mathbf{A}_{n}\bs_{n}-\frac{1}{p_{n}}\tr(\mathbf{A}_{n})\right|\ge\epsilon\ \right|\mathbf{A}_{n}\right] \\
& \leq \exp\left(-cp_n\epsilon^2/D^2\right) \\
 &  \sim \exp\left(-c\myc n\epsilon^2/D^2\right).
\end{align*}
The result follows.

\end{proof}

The matrices $\bfrh_n^{-1}$ and $\bfrh_n^{-1}\bfr_n\bfrh_n^{-1}$  are certainly independent of $\bs_n$.  The proof of Lemma~\ref{lem:trace} is complete if we observe that essential ranges of the smallest eigenvalues of these two matrices do not include zero.

\section{Maximization of \eqref{eq:eta-tilde-D}} \label{sec:eta-tilde-D}
For fixed $\bfr_n$, the quantity under consideration depends only on
\begin{equation} \label{eq:eta-tilde-D-2}
\frac{\tr(\bD_n^{-1})^2}{\tr\left(\bD_n^{-2}\bU_n'\bfr_n\bU_n\right)}.
\end{equation}
Let us abbreviate $p_n$ by $p$, $\bD_n$ by $\bD$, and $\bU'_n \bfr_n \bU_n$ by $\mathbf{C}$.  Then \eqref{eq:eta-tilde-D-2} is equivalent to 
\begin{equation}
\tr(\bD^{-1})^2/\tr(\bD^{-2}\mathbf{C}).
\label{eq:eta-tilde-D-3}
\end{equation}
  Observe that $\mathbf{C}$ is positive-definite.
If $C_{ii}$ are the diagonal entries of $\mathbf{C}$ and $D_{ii}$ are the diagonal entries of $\bD$, then we have
\[
\tr(\bD^{-1})^2 = \left( \sum_{i=1}^p ( D^{-1}_{ii} C_{ii}^{-1} )C_{ii} \right)^2.
\]
By Cauchy-Schwarz applied to the inner product $((a_i),(b_i))\in\mathbb{R}^p\times \mathbb{R}^p \mapsto \sum_{i=1}^p a_i b_i C_{ii}$, the above is bounded by\\
\[
\sum_{i=1}^p D_{ii}^{-2} C_{ii} \sum_{i=1}^p C_{ii}^{-1} = \tr(\bD^{-2} \mathbf{C})  \sum_{i=1}^p C_{ii}^{-1} , 
\]
with equality if and only if $D_{ii}^{-1}C_{ii} = 1$ for all $i$.  Thus, \eqref{eq:eta-tilde-D-3} is bounded above by $ \sum_{i=1}^p C_{ii}^{-1} $ with equality iff there is a positive scalar $\alpha$ such that $D_{ii}=\alpha C_{ii}$.  The proof is complete by noting that $C_{ii} =(\bU_n'\bfr_n\bU_n)_{ii}$ is equal to $\bu_{n,i}'\bfr_n \bu_{n,i}$.

\section{Proof of Theorem~\ref{thm:S-consistency-1}} \label{sec:lw-consistency}
Let the eigenvalues of $\bfr_n$ be $\tau_{n,1} \le\dots \le\tau_{n, p_n}$, and let $\boldsymbol{\lambda}_n = (\lambda_{n,1},  \lambda_{n,2}, \dots ,\lambda_{n,p_n})$ be the eigenvalues of $\bS_n$ in ascending order.  Let $\bU_n = [\bu_{n,1}, \dots, \bu_{n,p_n}]$ be a random unitary matrix satisfying the constraint that $\bU_n' \bS_n \bU_n$ is diagonal.
Let $[y]^+$ be defined as $\max\{y, 0\}$ and $h_n$ be defined as $n^{-0.35}$.
Temporarily suppressing the primary subscripts of the $\lambda_{n,j}$'s, we write $\boldsymbol{\lambda} = (\lambda_1, \dots \lambda_{p_n})$. Then $a(\lambda,\boldsymbol{\lambda})$ is defined by 
\[
\sum_{j=p-n+1}^{p}\frac{\mathrm{sgn}(\lambda-\lambda_{j})\sqrt{\left[(\lambda-\lambda_{j})^{2}-4\lambda_{j}^{2}h_{n}^{2}\right]^{+}}-\lambda+\lambda_{j}}{2 \lambda_{j}^{2}h_{n}^{2}}
\]
and $b(\lambda, \boldsymbol{\lambda})$ is defined by
\[
\sum_{j=p-n+1}^{p}\frac{\sqrt{\left[4\lambda_{j}^{2}h_{n}^{2}-(\lambda-\lambda_{j})^{2}\right]^{+}}-\lambda+\lambda_{j}}{2 \lambda_{j}^{2}h_{n}^{2}},
\]where the summands are defined to be zero when $j$ is not positive.  With $z_{n,j}= \pi \min\{n,p_n\}^{-1}(a(\lambda_{n,j},\boldsymbol{\lambda}_n)+ib(\lambda_{n,j},\boldsymbol{\lambda}_n))$, the shrunken eigenvalues $\td_{n,j}$ of \cite{ledoit2017direct} are
\begin{equation} \label{eq:dtildes}
\td_{n,j} \coloneqq \begin{cases}
\frac{\lambda_{n,j}}{|1-p/n-p/n\lambda_{n,j} z_{n,j} |^2}, & \text{if $\lambda_{n,j} > 0 $} \\
\frac{1}{\pi (p/n-1)a(0,\boldsymbol{\lambda}_n)/n}, & \text{if $\lambda_{n,j} = 0$}
\end{cases}
\end{equation}
Define $\check{d}_{n,j}$ by $\check{d}(\td_{n,j})$, where
\[
\check{d}(x) = \begin{cases}
\lambda_{n,p_n}, & \text{if $x> \lambda_{n,p_n}$} \\
1, & \text{if $x < 1$} \\
x, & \text{else}.
\end{cases}
\]
For $\bx\in\mathbb{R}^p$, let $\hat{\bx} = \operatorname{PAV}(\bx)$ be defined by the Pool-Adjacent Violators algorithm of \cite{ayer1955empirical}:
\[
\hat{\mathbf{x}} = \argmin\limits_{y_1\le y_2\le \dots \le y_p} \sum_{i=1}^p (x_i - y_i)^2.
\]
Let $\boldsymbol{\hat{d}}_n = \operatorname{PAV}(\boldsymbol{\check{d}}_n)$.
Then we claim that $\bfrh_n$ defined to be $\bU_n \mathrm{diag}(\hat{d}_{n,1}, \hat{d}_{n,2},\dots, \hat{d}_{n,p_n}) \bU_n'$ is oracle consistent.
\begin{rem}
We note that $\bU_n \mathrm{diag}(\boldsymbol{\check{d}}_n) \bU_n'$ is oracle-consistent as well, but appears to converge more slowly to the oracle for the realistic covariance considered in this paper.
\end{rem}

Let $H_n$ be the ``empirical spectral distribution function'' of $\bfr_n$:
 \[
 H_n(\tau) = p_n^{-1} \#\{ j: \tau_{n,j} \le \tau \},
 \]
 where $\tau_{n,1}\le \tau_{n,2}\le \dots \tau_{n,p_n}$ are the eigenvalues of $R_n$.
 Then 
 \[
 H_n(\tau) \to H(\tau) := \mathbf{1}(1\le \tau).
 \]
The empirical spectral distribution function of $\bS_n$ is
 \[
 F_n(\lambda) = p_n^{-1}\#\{j:\lambda_{n,j}\le \lambda\}.
 \]
 This is a random variable for each fixed $\lambda$. However, by the well-known Mar\v{c}enko-Pastur theorem \cite{marvcenko1967distribution} there is a deterministic c.d.f. $F$ such that $F_n(\lambda) \toas F(\lambda)$ for every point $\lambda$ at which $F$ is continuous \cite{marvcenko1967distribution, silverstein1995strong}.

 By \cite[Theorem~1.4]{ledoit2011eigenvectors}, there exists an integrable function $\delta: \mathbb{R}_{\ge 0} \to (0,\infty]$ such that
 \begin{equation} \label{eq:lp4}
 p_n^{-1} \sum_{i=1}^{p_n} \bu_{n,i}' \bfr_n \bu_{n,i} \mathbf{1}(\lambda_{n,i} \le \lambda) \toas \int_{-\infty}^\lambda \delta(l)\, dF(l),
 \end{equation}
 where $\mathbf{1}$ denotes an indicator function. 
 By \cite{ledoit2011eigenvectors, ledoit2017direct}, the function $\delta(l)$, which depends only on $F$ and $l$, is continuous except possibly where it is infinite. 
 Let $\delta_{n,i} \coloneqq \delta(\lambda_{n,i})$.

The shrunken eigenvalues $\td_{n,i}$ of \eqref{eq:dtildes} were shown in \cite{ledoit2017direct}  to have a key uniform consistency property: $\sup_{1\le i \le p_n} | \td_{n,i} - \delta_{n,i} | \toprob 0$ as $n\to \infty$.  Defining $\check{d}_{n,i} = \check{d}(\td_{n,i})$ and $\check{\delta}_{n,i} = \check{d}(\delta_{n,i})$, and using the fact that $ \sup_{1\le i \le p_n} | \check{d}_{n,i} - \check{\delta}_{n,i}| \le \sup_{1\le i\le p_n} | \tilde{d}_{n,i} - \delta_{n,i} |$, 
we can also say that 
\begin{equation} \label{eq:unif-consistency-ours}
 \sup_{1\le i \le p_n} | \check{d}_{n,i} - \check{\delta}_{n,i}|  \toprob 0,
\end{equation}
as $n\to\infty$.

Now consider $\boldsymbol{\hat{d}}_n = \operatorname{PAV}(\boldsymbol{\check{d}}_n)$.  Each entry $\hat{d}_{n,i}$ is bounded below by 1 and above by $\lambda_{n,p_n}$.  Further, since $\lambda_{n,p_n} \le M \left\Vert \bZ_n\bZ_n'/n\right\Vert$, with $\bZ_n = \bfr_n^{-1/2}\bX_n$, we have by \cite{bai2008limit} almost surely
\begin{equation} \label{eq:boundonlambda}
\limsup_{n\to\infty} \lambda_{n,p_n} \le M (1+\sqrt{\gamma})^2,
\end{equation}
where we recall that $M = \left\Vert \bfr_n \right\Vert$.
Thus, the estimator $\bfrh_n$ satisfies condition (ii) of Definition~\ref{def:shrinkage}.  From \eqref{eq:boundonlambda} and Cauchy-Schwarz that $\bfrh_n$ is oracle-consistent if 
\[
\frac{1}{p_n} \sum_{i=1}^{p_n} |\hat{d}_{n,i} - \bu_{n,i}' \bfr_n \bu_{n,i} | \toprob 0.
\]
Thus, for oracle consistency, it suffices to show
 \begin{equation} \label{eq:hat-consistency}
 p_n^{-1} \sum_{i=1}^{p_n} \left| \hat{d}_{n,i} - \check{\delta}_{n,i} \right| \toprob 0
 \end{equation}
 and
\begin{equation} \label{eq:oracles-consistency}
 p_n^{-1} \sum_{i=1}^{p_n} \left| \check{\delta}_{n,i} - \bu_{n,i}'\bfr_n\bu_{n,i} \right| \toas 0.
 \end{equation}

Consider \eqref{eq:oracles-consistency}.
%
%
We may expand $d_{n,i}^* = \bu_{n,i}'\bfr_n \bu_{n,i}$ in terms of the eigenvectors $\bv_{n,j}$ of $\bfr_n$ having eigenvalues $\tau_{n,j}$:
\begin{align*}
& \bu_{n,i}' \bfr_n \bu_{n,i}  \\
&  = \sum_{j=1}^{p_n} \tau_{n,j} \left| \langle \bu_{n,i}, \bv_{n,j} \rangle \right|^2 \\
& = \sum_{j=p_n-r+1}^{p_n} \tau_{n,j} \left| \langle \bu_{n,i}, \bv_{n,j} \rangle \right|^2 + \sum_{j=1}^{p_n-r} \left| \langle \bu_{n,i}, \bv_{n,j} \rangle \right|^2 \\
& = 1 + \sum_{j=p_n-r+1}^{p_n} (\tau_{n,j} - 1) \left| \langle \bu_{n,i}, \bv_{n,j} \rangle \right|^2,
\end{align*}
where we have used the identity $\left(\sum_{j=1}^{p_n - r} +\sum_{j=p_n-r+1}^{p_n} \right) \left| \langle \bu_{n,i}, \bv_{n,j} \rangle \right|^2 = 1$
Thus, for \eqref{eq:oracles-consistency}, it suffices to show
\begin{equation} \label{eq:deltas}
p_n^{-1} \sum_{i=1}^{p_n} \left| \check{\delta}_{n,i} - 1 \right| \toas 0
\end{equation}
and
 \begin{equation} \label{eq:biases}
(M - 1) p_n^{-1}  \sum_{i=1}^{p_n}   \sum_{j=p_n-r+1}^{p_n} \left| \langle \bu_{n,i}, \bv_{n,j} \rangle\right|^2 \toas 0.
\end{equation}

Consider \eqref{eq:biases}.  The result \cite[Theorem~1.3]{ledoit2011eigenvectors} states that there is an integrable function $\varphi(l,t)$ such that
\begin{align*}
& \Phi_{p_n}(\lambda, \tau) \\
& \coloneqq \frac{1}{p_n} \sum_{j=1}^{p_n} \sum_{i=1}^{p_n} \left| \langle \bu_{n,i}, \bv_{n,j} \rangle\right|^2 \mathbf{1}(\lambda_{n,i} \le \lambda) \mathbf{1}(\tau_{n,j} \le \tau)
\end{align*}
converges almost surely to
\[
\int_{-\infty}^\lambda \int_{-\infty}^\tau \varphi(l,t)\, dH(t)dF(l).
\]
Thus, the left side of \eqref{eq:biases} is proportional to $\Phi_{p_n}(\max\mathrm{supp} F, \tau_{p_n}) - \Phi_{p_n}(\max\mathrm{supp} F, \tau_{p_n-r+1})$, which, because the set $\{\tau_{p_n-r+1}, \dots \tau_{p_n}\}$ is outside the support of $dH$, converges almost surely to zero, as desired.

Consider \eqref{eq:deltas}. As we have stated, $\delta$ is continuous except possibly where it is infinite.  
Letting $\check{\delta}(\lambda) = \check{d}(\delta(\lambda))$, this problem of infinities is removed if we restrict to a null-complemented event where for $n$ large enough $\lambda_{n,p_n} \le M(1+\sqrt{\gamma})^2+\epsilon$.
Thus, for \eqref{eq:deltas}, by the portmanteau theorem
\begin{align}
\frac{1}{p_n} \sum_{i=1}^{p_n} \left| \check{\delta}_{n,i} - 1 \right| & = \int \left| \check{\delta}(l) - 1 \right|\, dF_n(l) \nonumber \\
& \toas \int \left| \check{\delta}(l) - 1 \right|\, dF(l). \label{eq:desired-lp-zero}
\end{align}
Since $\delta$ depends only on the limiting spectral distribution function $F$, it is the same for the sequence $\bfr_n = \mathbf{I}_{p_n}$ as for our sequence.  For this simpler sequence we have
\begin{align*}
F_n(\lambda) & = \frac{1}{p_n} \# \{ i : \lambda_{n,i} \le \lambda \} \\
& = \frac{1}{p_n} \sum_{i=1}^{p_n} \bu_{n,i}' \mathbf{I}_{p_n} \bu_{n,i} \mathbf{1}(\lambda_{n,i}\le \lambda) \\
& \toas \int_{-\infty}^\lambda \delta(l)\, dF(l),
\end{align*}
where the convergence follows from \cite[Theorem~1.4]{ledoit2011eigenvectors}.  Since $F_n$ also converges weakly almost surely to $F$, it must be that
\[
F(\lambda) = \int_{-\infty}^\lambda \delta(l)\, dF(l),
\]
which by the fundamental theorem of calculus, implies that $\delta$ is $F$-a.e. equal to 1. Thus, so is $\check{\delta}$, as desired.

We now prove \eqref{eq:hat-consistency}.  By Cauchy-Schwarz, this relation follows if we can show that $g_n(\bdhat_n, \bdeltacheck_n) \toprob 0$, where $g_n(\ba_n, \bb_n)$ is defined for $\ba_n$ and $\bb_n$ in $\mathbb{R}^{p_n}$ as
\[
p_n^{-1} \sum_{i=1}^{p_n} (b_{n,i} - a_{n,i})^2.
\]
With this definition $\boldsymbol{\hat{a}}_n \coloneqq \operatorname{PAV}(\ba_n)$ satisfies
\[
 \boldsymbol{\hat{a}}_n = \argmin\limits_{y_{n,1} \le y_{n,2} \le \dots \le y_{n,p_n}} g_n(\ba_n, \mathbf{y}_n).
\]
Consider first the lemma below.
\begin{lem} \label{lem:closest-points}
 Suppose $\boldsymbol{a}_{n}$ and $\boldsymbol{b}_{n}$ are sequences of random $p_n$-vectors whose components lie between 1 and $\lambda_{n,p_n}$ and such that $g_n(\ba_n, \bb_n)$ converges in probability to zero as $n\to\infty$.  Then
 \begin{align} \label{eq:lem-ineq}
& \left| g_n(\ba_n, \boldsymbol{\hat{a}}_n) - g_n(\bb_n, \boldsymbol{\hat{b}}_n) \right| \toprob 0. 
 \end{align}
 \end{lem}
 
\begin{proof}
Let $\epsilon_n = g_n(\ba_n, \bb_n)^{1/2}$. Suppose $\by_n$ has components between 1 and $\lambda_{n,p_n}$.  
Then we have
\begin{align}
& \left| g_n(\boldsymbol{\hat{b}}_n, \by_n) - g_n(\boldsymbol{\hat{a}}_n, \by_n) \right| \nonumber \\ 
& \le 4\lambda_{n, p_n} p_n^{-1}\sum_{i=1}^{p_n} |a_{n,i} - b_{n,i} | \nonumber \\
& \le 4\lambda_{n,p_n} \epsilon_n, \label{eq:diff-of-gs}
\end{align}
where the last inequality follows from Cauchy-Schwarz. Thus,
\begin{align*}
g_n(\ba_n,\boldsymbol{\hat{a}}_n) & \le g_n(\ba_n, \boldsymbol{\hat{b}}_n) \\
& \le g_n(\bb_n, \boldsymbol{\hat{b}}_n) + 4\lambda_{n,p_n}\epsilon_n \end{align*}
The reverse inequality is proved in the same way, and the result follows since $\epsilon_n \toprob 0$.
\end{proof}

Again using the fact that, on a null-complemented event, $\check{\delta}$ is equal to 1 on $\mathrm{supp}F$ and continuous and bounded elsewhere
 we have that
$$ p_n^{-1} \sum_{i=1}^{p_n} (\check{\delta}_{n,i} - 1)^2 \toas \int (\check{\delta}(\lambda)-1)^2\, dF(\lambda) = 0,$$
where the limiting statement follows from the portmanteau theorem.
Thus, $$ g_n((1,1,\dots, 1), \bdeltacheck_n) \toprob 0,$$ and so does  smaller quantity $g_n(\bdeltahat_n, \bdeltacheck_n)$.

We now show that $g_n(\bdhat_n, \bdeltacheck_n) \toprob 0$ as $n\to\infty$, as promised.  By \eqref{eq:unif-consistency-ours} and \eqref{eq:diff-of-gs}, we have $g_n(\bdhat_n, \bdeltacheck_n) - g_n(\bdhat_n, \bdcheck_n) \toprob 0$.  By the Lemma~\ref{lem:closest-points}, with $\ba_n = \bdcheck_n$ and $\bb_n = \bdeltacheck_n$, we get $g_n(\bdhat_n, \bdcheck_n) - g_n(\bdeltahat_n, \bdeltacheck_n) \toprob 0$. But, as indicated by the last paragraph, $g_n(\bdeltahat_n, \bdeltacheck_n) \toas 0$, so the proof is complete.

%
 
\section{Proof of Theorem~\ref{thm:nsnr-consistent}}  \label{sec:ab-lemma}
First we prove a lemma.  Recall that the shrinkage oracle $\bfr_n^*$ is defined by $\bU_n \mathrm{diag}(d_{n,1}^*, d_{n,2}^*, \dots, d_{n, p_n}^*) \bU_n'$, where $\bU_n = [\bu_{n,1}, \bu_{n,2},\dots, \bu_{n,p_n}]$ is a matrix of column eigenvectors of sample covariance and $d_{n,i}^* = \bu_{n,i}' \bfr_n \bu_{n,i}$.  Because its eigenvalues are the same as sample covariance's, it commutes with any shrinkage estimator.
\begin{lem} \label{lem:ab-lem}
For each $n$, let $\bfrh_n$ be oracle consistent and let $\mathbf{M}_n$ be a random matrix such that $\sup_n \left \Vert \mathbf{M}_n \right \Vert$ is almost surely bounded. Then
\begin{equation} \label{eq:a-power-convergence}
\frac{1}{p_n} \left|\tr(\bfrh_n^a \mathbf{M}_n) - \tr(\bfr_n^{*a} \mathbf{M}_n ) \right| \toprob 0
\end{equation}
for all negative integers $a$.
\end{lem}
\begin{proof}

\end{proof}

We now prove Theorem~\ref{thm:nsnr-consistent}.  First, we note the conclusion of Lemma~1 holds with $\bfrh_n$ replaced by $\bfrs_n$ because the latter matrix is independent of $\bs_n$ and satisfies the conditions of Lemma~\ref{lem:trace}. 
By this extension of Lemma~\ref{lem:trace-approx-eta}, then, it suffices to show 
\[
	\left \vert \tilde{\eta}(\bfrh_n, \bfr_n)  - \tilde{\eta}(\bfrs_n, \bfr_n) \right \vert \toprob 0.
	\]
Using the circulant property of trace and the fact that $p_n^{-1} \tr(\bfr_n^{-1}) \to 1$, the above is equivalent to the statement that
\[
	\left| \frac{X_n}{Y_n} - \frac{\tilde{X}_n}{\tilde{Y}_n} \right| \toprob 0,
	\]
where $X_n = p_n^{-2}\tr(\bfrh_n^{-1})^2$, $Y_n=p_n^{-1} \tr(\bfrh_n^{-2}\bfr_n)$, $\tilde{X}_n =p_n^{-2}\tr(\bfr_n^{*-1})^2$, and $\tilde{Y}_n=p_n^{-1} \tr(\bfr_n^{*-2}\bfr_n)$.  The fact that $|X_n-\tilde{X}_n |\toprob 0$ and $|Y_n-\tilde{Y}_n |\toprob 0$ follows from Lemma~\ref{lem:ab-lem}. The result then follows from Lemma~\ref{lem:quotients} with $f(x,y)=x/y$.

\section{Proof of Lemma~\ref{lem:limit-of-xi}} \label{sec:limit-of-xi}
	By Lemma~\ref{lem:trace} and Lemma~\ref{lem:quotients} with $f(x,y)=x/y$, and by the definition of $\xi_n$ given in Lemma~\ref{lem:limit-of-xi}, we have
\begin{equation} \label{eq:triangle-ineq-1}
	\left \vert \xi_n - \frac{\tr\left(\bfrh_n^{-2}\bfr_n \right)}{\tr\left(\bfrh_n^{-1}\right)} \right \vert \overset{a.s.}{\to} 0.
	\end{equation}
Furthermore, by Lemma~\ref{lem:ab-lem} we have
\[
	\frac{1}{p_n} \left| \tr(\bfrh_n^{-2} \bfr_n) - \tr(\bfr_n^{*-2}\bfr_n) \right| \toprob 0
	\]
and
\[
	\frac{1}{p_n} \left| \tr(\bfrh_n^{-1}) - \tr(\bfr_n^{*-1}) \right| \toprob 0.
	\]
Using Lemma~\ref{lem:quotients} again, the last two equations imply
\begin{align} \label{eq:triangle-ineq-2}
\left \vert \frac{\tr\left(\bfrh_n^{-2}\bfr_n\right)}{\tr\left(\bfrh_n^{-1}\right)} - \frac{\tr\left(\bfr_n^{*-2}\bfr_n\right)}{\tr\left(\bfr_n^{*-1}\right)} \right \vert \toprob 0.
\end{align}
Combining \eqref{eq:triangle-ineq-1} and \eqref{eq:triangle-ineq-2}, we get
\begin{align*}
\left \vert \xi_n- \frac{\tr\left(\bfr_n^{*-2}\bfr_n\right)}{\tr\left(\bfr_n^{*-1}\right)} \right \vert \toprob 0.
\end{align*}
The quantity
\[
	\frac{\tr\left(\bfr_n^{*-2}\bfr_n\right)}{\tr\left(\bfr_n^{*-1}\right)}
	\]
 is identically 1 by using the identity
\begin{align}
	{\tr\left(\bfr_n^{*-2}\bfr_n\right)} & =\tr(\mathbf{D}^{*-2}_n \mathbf{U}_n' \bfr_n \mathbf{U}_n) \nonumber \\
		& = \tr(\mathbf{D}^{*-1}_n) \nonumber \\
		& = \tr(\bfrsi_n), \label{eq:nostars}
	\end{align}
yielding the desired result.

\section{Proof of Lemma~\ref{thm:sinr}} \label{sec:pf-thm-sinr}

We claim that
\[
\pi_n(t) := \Pr[|\hat{\nu}_n(\bx)-\nu_n|\ge t \mid \fH_1^n, \bs_n, \bX_n]
\]
satisfies  $\pi_n(t) \lesssim e^{-t^2}$, where as before, if $X_n$ and $Y_n$ are random variables, $X_n \lesssim Y_n$ means $\max\{X_n - Y_n, 0\}$ converges in probability to 0 as $n\to\infty$.
By the continuity of $t\mapsto e^{-t^2}$, it suffices to show $\pi_n(t) \lesssim e^{-(t-\delta)^2}$ for every $\delta > 0$.  

Our approach is to find $\hat{\mu}_n(\bx)$ and $\mu_n$---approximations of $\hat{\nu}_n(\bx)$ and $\nu_n$, resp.---such that the expression $\tilde{\pi}_n(t)$ defined by
\[
\tilde{\pi}_n(t) := \Pr[|\hat{\mu}_n(\bx) - \mu_n| \ge t \mid \fH_1^n, \bs_n, \bX_n]
\]
satisfies 
\begin{equation} \label{eq:pi-convergence}
\pi_n(t) \lesssim \tilde{\pi}_n(t-\delta),
\end{equation}
for all $\delta> 0$, and satisfies
 \begin{equation} \label{eq:pitilde-bound}
 \tilde{\pi}_n(t) \le e^{-t^2}.
 \end{equation}
 To this end, let us make the definitions
 \[
\mu_n \coloneqq \nu(\bs_n,\bfrs_n,\bfr_n) = |a| \frac{\bs_n'\bfr_n^{*-1}\bs_n}{(\bs_n'\bfr_n^{*-1}\bfr_n\bfr_n^{*-1}\bs_n)^{1/2}}
 \]
 and
\[
\hat{\mu}_n(\bx)\coloneqq  \frac{|\bs'_n \bfr_n^{*-1} \bx|}{\left(\bs_n' \bfr_n^{*-1} \bfr_n \bfr_n^{*-1} \bs_n\right)^{1/2}}.
\]
The inequality \eqref{eq:pitilde-bound} follows immediately since, by the triangle inequality,
  \begin{align*}
 & \left|\hat{\mu}_n(\bx)-\mu_n \right|\\
  & \le \left|\frac{\bs_n'\bfr_n^{*-1}(\bx-a\bs_n)}{(\bs_n'\bfr_n^{*-1}\bfr_n\bfr_n^{*-1}\bs_n)^{1/2}}\right|\\
  & = \frac{|\bs_n'\bfr_n^{*-1}\bd|}{(\bs_n'\bfr_n^{*-1}\bfr_n\bfr_n^{*-1}\bs)^{1/2}}=:|Z|,
\end{align*}
which is the modulus of a (complex, circularly symmetric) Gaussian-distributed
random variable $Z$ with mean zero
and variance one. 

For \eqref{eq:pi-convergence}, we need two lemmas.
\begin{lem} \label{lem:numu}
With $\nu_n$ and $\mu_n$ as above, $|\nu_n - \mu_n| \toprob 0$.
\end{lem}
\begin{proof}
Using Lemma~\ref{lem:trace} twice and Lemma~\ref{lem:ab-lem} once, $X_n$ defined to be $\bs_n'\bfrh_n^{-1}\bs_n$ satisfies $|X_n - \tilde{X}_n| \toprob 0$, where $\tilde{X}_n = \bs_n'\bfr_n^{*-1}\bs_n$.  Similarly $Y_n$, defined to be $\bs_n'\bfrh_n^{-1}\bfr_n\bfrh_n^{-1}\bs_n$ satisfies $|Y_n-\tilde{Y}_n|\toprob 0$, where $\tilde{Y}_n = \bs_n'\bfr_n^{*-1}\bfr_n\bfr_n^{*-1}\bs_n$.  Applying Lemma~\ref{lem:quotients} with $f(x,y)=x/y^{1/2}$, we have $|X_n/Y_n^{1/2}-\tilde{X}_n/\tilde{Y}_n^{1/2}| \toprob 0$, as desired.
\end{proof}

The second lemma is as follows.
\begin{lem} \label{lem:numuhat}
Let $\epsilon_n(\bx) = |\hat{\nu}_n(\bx)-\hat{\mu}_n(\bx)|$ and let $\delta_n \toprob 0$. Then, for all $\delta>0$, we have
\[
\tilde{\epsilon}_n(\delta) \coloneqq \Pr[\epsilon_n(\bx) +\delta_n \ge \delta \mid \fH_1^n, \bs_n, \bX_n] \toprob 0.
\]
\end{lem}
\begin{proof}
The quantity $\epsilon_n(\bx)$ is equal to
\[
\left| |\hat{\mathbf{w}}_n' \bx|-|(\mathbf{w}^*_n)'\bx|\right|,
\]
where
\[
\hat{\mathbf{w}}_n' = \frac{\bs_n'\bfrh_n^{-1}}{\eta_n^{1/2}}\qquad\text{and}\qquad(\mathbf{w}^*_n)' = \frac{\bs_n'\bfr_n^{*-1}}{\eta_n^{*1/2}}
\]
and
\[
\eta_n = \bs_n'\bfrh_n^{-1}\bs_n\qquad\text{and}\qquad \eta^*_n = \bs_n \bfr_n^{*-1}\bfr_n\bfr_n^{*-1}\bs_n.
\]
Thus,
\begin{equation} \label{eq:bound-for-epsilon-n}
\epsilon_n(\bx) \le |(\hat{\mathbf{w}}_n-\mathbf{w}^*_n)'\bx| = |\bs_n'\mathbf{P}_n\bx|,
\end{equation}
where
\[
\mathbf{P}_n = \frac{\bfrh^{-1}_n}{(\bs_n'\bfrh_n^{-1}\bs_n)^{1/2}} - \frac{\bfr^{*-1}_n}{\left(\bs_n \bfr^{*-1}_n \bfr_n \bfr^{*-1}_n \bs_n\right)^{1/2}}.
\]
For fixed $\bs_n,\bfrh_n$, and $\bfr_n$, the right side of \eqref{eq:bound-for-epsilon-n} is the absolute value of the random variable $V_n=\bs'_n\mathbf{P}_n\bx$---conditionally Gaussian given $\bs_n$ and $\bX_n$---which has conditional mean $a\bs_n'\mathbf{P}_n\bs_n$ and conditional variance $v_n = \bs_n'\mathbf{P}_n\bfr_n \mathbf{P}_n\bs_n$,.  Thus, we have that $|V_n|/v_n^{1/2}$ is non-central Chi-distributed with noncentrality parameter
\[
\lambda_n = |a|^2(\bs_n\mathbf{P}_n\bs_n)^2/v_n.
\]
We have that
\begin{align*}
& \tilde{\epsilon}_n(\delta) \\
 & \le \Pr[|V_n|+\delta_n \ge \delta \mid \fH_1^n, \bs_n, \bX_n ] \\ 
& = \Pr[|V_n|^2/v_n\ge (\max\{(\delta-\delta_n),0\})^2/v_n \mid \fH_1^n, \bs_n, \bX_n] \\
& = Q((\max\{\delta-\delta_n,0\})^2/v_n, \lambda_n).
\end{align*}
Further, by the triangle inequality; the boundedness of $\sup_n \left\Vert \bfrh_n^{-1}\right\Vert$, $\left\Vert\bfr_n^{\pm 1}\right\Vert$, and $\left\Vert \bfr_n^{*\pm 1}\right\Vert$; and the almost sure boundedness of $\limsup_n\left\Vert\bfrh_n\right\Vert$, we have that $\lambda_n$ is bounded above by a constant $C>0$ for $n$ sufficiently large.  Thus, $\tilde{\epsilon}_n(\delta)$ is bounded above according to
\[
\tilde{\epsilon}_n(\delta) \le Q((\max\{\delta-\delta_n,0\})^2/v_n,C).
\]
Let $\epsilon > 0$. Since $\delta_n\toprob 0$, taking $n$ large enough, with high probability we have $\delta-\delta_n \ge \delta/2$. Then $\tilde{\epsilon}_n(\delta) \le Q(\delta^2 /(4v_n), C)$.  If $v_n \toprob 0$, then by the continuous mapping theorem and the fact that $\lim_{\tau \to\infty} Q(\tau,C)= 0$, the proof will be complete.

We will thus aim to show that $v_n\toprob 0$. We have that $v_n$ is equal to the sum of
\begin{equation} \label{eq:first-half-V}
\frac{\bs'_n\bfrh_n^{-1}\bfr_n\bfrh_n^{-1}\bs_n}{\eta_n} - \frac{\bs_n'\bfr_n^{*-1}\bfr_n\bfrh_n^{-1}\bs_n}{\eta_n^{1/2}\eta_n^{*1/2}}
\end{equation}
and
\begin{equation} \label{eq:second-half-V}
\frac{\bs'_n\bfr_n^{*-1}\bfr_n\bfr_n^{*-1}\bs_n}{\eta^*_n} - \frac{\bs_n'\bfr_n^{*-1}\bfr_n\bfrh_n^{-1}\bs_n}{\eta_n^{1/2}\eta_n^{*1/2}}
\end{equation}
By Lemma~\ref{lem:trace} twice and the triangle inequality, \eqref{eq:first-half-V} converges to zero if
\begin{equation} \label{eq:different-nums}
\left| \frac{p_n^{-1}\tr(\bfrh_n^{-2}\bfr_n)}{\eta_n} - \frac{p_n^{-1} \tr(\bfr_n^{*-2}\bfr_n)}{\eta_n^{1/2}\eta_n^{*1/2}}\right| \toprob 0.
\end{equation}
Let $X_n = p_n^{-1}\tr(\bfrh_n^{-2}\bfr_n)$ and $\tilde{X}_n = p_n^{-1}\tr(\bfr_n^{*-2}\bfr_n)$ and $Y_n= \tilde{Y}_n = \eta_n$ and $Z_n = \eta_n$ and $\tilde{Z}_n = \eta^*_n$. 
By Lemma~\ref{lem:trace}, $|Z_n - p_n^{-1}\tr(\bfrh_n^{-1})| \toprob 0$.  By Lemma~\ref{lem:ab-lem}, then,
$|Z_n - p_n^{-1}\tr(\bfr_n^{*-1})| \toprob 0$.  By Lemma~\ref{lem:trace}, $|\tilde{Z}_n - p_n^{-1}\tr(\bfr_n^{*-1}\bfr_n\bfr_n^{*-1})|\toprob 0$. By \eqref{eq:nostars} the latter trace is identically equal to $\tr(\bfr_n^{*-1})$.  Thus, $|Z_n-\tilde{Z}_n| \toprob 0$.  By Lemma~\ref{lem:quotients} with $f(x,y,z)=x/(y^{1/2}z^{1/2})$, $|f(X_n,Y_n,Z_n) - f(\tilde{X}_n,\tilde{Y}_n,\tilde{Z}_n)| \toprob 0$, proving \eqref{eq:first-half-V} converges to 0 in probability.  The proof for \eqref{eq:second-half-V} is similar.  Thus, we have proven $v_n\toprob 0$, as desired. 
\end{proof}

We may now prove  \eqref{eq:pi-convergence}.  Let $\delta>0$ and let $\delta_n = |\nu_n-\mu_n|$.  We have
\begin{align*}
\pi_n(t) & \le \Pr[\epsilon_n(\bx)+\delta_n+|\hat{\mu}_n(\bx)-\mu_n|\ge t\mid \fH_1^n, \bs_n, \bX_n] \\
& \le \epsilon'_n + \tilde{\pi}_n(t-\delta),
\end{align*}
where $\epsilon'_n=\Pr[\epsilon_n(\bx)+\delta_n \ge \delta \mid \fH_1^n, \bs_n, \bX_n] $.  Subtracting and taking the max with 0 gives
\[
\max\{\pi_n(t)- \tilde{\pi}_n(t-\delta),0\} \le \epsilon'_n.
\]
Thus, by Lemmas~\ref{lem:numu} and \ref{lem:numuhat}, $\pi_n(t)$ is less than or asymptotically equal to $\tilde{\pi}_n(t-\delta)$ and the proof is complete.

\section{Proof of Theorem~\ref{thm:pd}} \label{sec:pd}

Let $Q_\rho(\nu) = Q(\rho,\nu^2)$ for $\rho,\nu\ge 0$, and let $q_{n\pm}(\bx) = Q_\tau(\nu_{n\pm}(\bx,\myq))$. Recall that we can write $p_{\text{d}}^n(\tau)$ as $Q(\tau/\xi_n,\nu_n^2)=Q_{\tau/\xi_n}(\nu_n)$. We thus wish to show that
\[
\pi_n \coloneqq \Pr\left[ q_{n-}(\bx) \le Q_{\tau/\xi_n}(\nu_n) \le q_{n+}(\bx)\mid \fH_1^n,\bs_n, \bX_n \right]
\]
is greater than or asymptotically equal to $\myq$. Equivalently, if we define $t$ to be $\sqrt{\log(1/(1-\myq))}$, we wish to show $\pi_n$ is greater than or asymptotically equal to $1-e^{-t^2}$.

Based on Definition~\ref{def:shrinkage}(ii), there is a constant $C > 0$ such that $\limsup_n \nu_n < C$ almost surely.  Given $\delta > 0$, this means that there is $n_0 = n_0(\delta)$ such that $\nu_n < C$ for $n\ge n_0$ except with probability at most $\delta$. Similarly, there is a constant $c > 0$ such that $\liminf_n \nu_n > c$ almost surely.  There is thus $n_0$ is large enough that
\begin{equation} \label{eq:nu-bounds}
c < \nu_n < C
\end{equation}
for $n\ge n_0$ except with probability at most $\delta$.  In order to prove the desired convergence in probability, we assume \eqref{eq:nu-bounds} holds for the rest of this proof.

First, we study $\tilde{\pi}_n(\epsilon)$, defined as
\[
\Pr\left[ q_{n-}(\bx)+\epsilon \le Q_\tau(\nu_n) \le q_{n+}(\bx) - \epsilon \mid\fH_1^n, \bs_n, \bX_n \right].
\]
It can be shown using elementary analysis that $Q_\tau^{-1}$, which is a mapping from $[Q_\tau(0),1)$ to $[0,\infty)$, is well-defined and Lipschitz on $(a,b)$ for $a>Q_\tau(0)$ and $b<1$. 
 Suppose $0< \epsilon < c/2$. Then by monotonicity of $Q_\tau$, the numbers $Q_\tau(\nu_n-\epsilon)$, $Q_\tau(\nu_n)$, and $Q_\tau(\nu_n+\epsilon)$ all lie in the interval $(Q_\tau(c/2), Q_\tau(C+c/2))$. Thus, there is a Lipschitz constant $L>0$ such that
 \begin{align*}
 \epsilon & = Q_\tau^{-1}Q_\tau(\nu_n+\epsilon)-Q_\tau^{-1}Q_\tau(\nu_n) \\
 & \le L(Q_\tau(\nu_n+\epsilon)-Q_\tau(\nu_n))
 \end{align*}
 and
  \begin{align*}
 \epsilon & = Q_\tau^{-1}Q_\tau(\nu_n)-Q_\tau^{-1}Q_\tau(\nu_n-\epsilon) \\
 & \le L(Q_\tau(\nu_n)-Q_\tau(\nu_n-\epsilon)).
 \end{align*}
 
Continue assuming that $0<\epsilon < c/2$. By the above the statement $q_{n-}(\bx)+\epsilon/L \le Q_\tau(\nu_n)$ is implied by $q_{n-}(\bx)  \le Q_\tau(\nu_n-\epsilon)$ and the statement $Q_\tau(\nu_n) \le q_{n+}(\bx) - \epsilon/L $ is implied by $Q_\tau(\nu_n+\epsilon) \le q_{n+}(\bx)$.  Thus, $\tilde{\pi}_n(\epsilon/L) $ is greater than or equal to the probability that
 \begin{align*}
q_{n-}(\bx)\le Q_\tau(\nu_n-\epsilon) \wedge Q_\tau(\nu_n+\epsilon) \le q_{n+}(\bx)
 \end{align*}
 given $\fH_1^n$ and $\bs_n$ and $\bX_n$.
 Using monotonicity of $Q_\tau$ again, the latter is equal to
 \begin{align*}
 & \Pr\left[ \nu_{n-}(\bx)\le \nu_n-\epsilon \wedge \nu_n+\epsilon \le \nu_{n+}(\bx)\mid\fH_1^n, \bs_n, \bX_n \right] \\
 & = \Pr\left[ \left| \nu_n(\bx)-\nu_n \right| \le t-\epsilon \mid \fH_1^n, \bs_n, \bX_n \right].
 \end{align*}
 But by Lemma~\ref{thm:sinr}, this means that $\tilde{\pi}_n(\epsilon/L)$ is   asymptotically greater than or equal to $1-e^{-(t-\epsilon)^2}$.  Thus,
 \[
 M_n \coloneqq \max\{1-e^{-(t-\epsilon)^2}-\tilde{\pi}(\epsilon/L),0\} \toprob 0
 \]
 as $n\to\infty$.
 
 Next consider $\tilde{M}_n \coloneqq \max\{ \tilde{\pi}_n(\epsilon/L) - \pi_n, 0\}$ with $L$ and $\epsilon$ as above. Using the formula $\Pr[A]-\Pr[B] \le \Pr[A \backslash B]$, and using elementary logical manipulations, we have that
 \[
 \tilde{M}_n \le \mathbf{1}\left( \left| Q_{\tau/\xi_n}(\nu_n) - Q_\tau (\nu_n) \right| > \epsilon/L \right),
 \]
 where $\mathbf{1}$ is the indicator function. 
 Using the fact that $I_0(z)$ is continuous and asymptotically equal to $e^z/\sqrt{2\pi z}$ as $z\to \infty$ \cite{abramowitz1988handbook}, 
 there is a constant $B>0$ such that 
 \[
 \left|Q_{\tau/\xi_n}(\nu_n) - Q_\tau(\nu_n)\right| \le B|\tau/\xi_n - \tau|.
 \]
 It then follows that
 \[ 
 \tilde{M}_n \le \mathbf{1}\left( B \left| \tau/\xi_n-\tau\right| > \epsilon/L\right).
 \]
 But this converges in probability to 0 by Lemma~\ref{lem:limit-of-xi} as $n\to\infty$. Thus, $\tilde{M}_n$ goes to zero in probability as $n\to\infty$.

Finally, consider $M_n^* \coloneqq \max\{1-e^{-t^2}-\pi_n,0\}$.  We have that
\[
0 \le M_n^* \le 1-e^{-t^2}-(1-e^{-(t-\epsilon)^2})+ M_n + \tilde{M}_n.
\]
Since the latter two terms converge in probability to zero and $\epsilon$ can be arbitrarily small, the result follows.

\section*{Acknowledgment}

This work was supported by the United States Air Force Sensors Directorate, AFOSR grant 19RYCOR036, and ARO grant W911NF-15-1-0479. However, the views and opinions expressed in this article are those of the authors and do not necessarily reflect the official policy or position of any agency of the U.S. government. Examples of analysis performed within this article are only examples. Assumptions made within the analysis are also not reflective of the position of any U.S. Government entity. The Public Affairs approval number of this document is 88ABW-2020-3037.

\ifCLASSOPTIONcaptionsoff
  \newpage
\fi



%
\bibliographystyle{plain}
\bibliography{RMH-STAP-bib}

%

%
%
%




\end{document}